\documentclass[sigconf,screen]{acmart}

\usepackage{amsthm,amsmath,booktabs,xcolor,graphicx}
\usepackage{thm-restate}
\usepackage{bbm}
\usepackage{listings}
\usepackage{xcolor}
\usepackage{appendix}
\usepackage{enumerate}
\usepackage{mathbbol}
\usepackage[utf8]{inputenc}
\usepackage[T1]{fontenc}
\usepackage{url}
\usepackage[disable]{todonotes}
\usepackage{etoolbox}
\usepackage{appendix}
\usepackage{xspace}
\usepackage[linesnumbered,ruled,vlined]{algorithm2e}
\usepackage[noend]{algpseudocode}
\usepackage[labelfont=bf]{caption}
\usepackage[noabbrev,capitalize]{cleveref}
\crefname{equation}{}{} 
\AtBeginEnvironment{appendices}{\crefalias{section}{appendix}} 

\usepackage[color,final]{showkeys} 

\colorlet{refkey}{orange!20}
\colorlet{labelkey}{blue!30}

\numberwithin{equation}{section}
\newtheorem{theorem}{Theorem}[section]

\crefname{claim}{Claim}{Claims}

\newtheorem*{question*}{Question}

\theoremstyle{definition}
\newtheorem{definition}[theorem]{Definition}

\newtheorem*{definition*}{Definition}

\theoremstyle{remark}

\newcommand{\alert}[1]{{\color{orange}#1}}
\renewcommand{\alert}[1]{}
\renewcommand{\tilde}{\widetilde}

\newcommand{\poly}{\mathrm{poly}}

\newcommand{\fC}{\mathcal{C}}

\newcommand{\eps}{\varepsilon}

\newcommand{\R}{\mathbb{R}}

\newcommand{\diam}{\textrm{diam}}

\newcommand{\fO}{\mathcal{O}}
\newcommand{\setssp}{set-source shortest path\xspace}
\newcommand{\tO}{\tilde{O}}

\newcommand{\congest}{\ensuremath{\mathsf{CONGEST}}\xspace}

\newcommand{\aggregate}{\ensuremath{\mathsf{Minor}}\text{-}\ensuremath{\mathsf{Aggregation}}\xspace}
\newcommand{\pram}{\ensuremath{\mathsf{PRAM}}\xspace}
\newcommand{\euler}{\ensuremath{\textsc{Eulerian-}\textsc{Orientation}}\xspace}

\newcommand{\oEuler}{\fO^{Euler}}
\newcommand{\oRound}{\fO^{Round}}
\newcommand{\oDist}{\fO^{Forest}}

\newcommand{\oPot}{\fO^{Pot}}
\newcommand{\oObliv}{\fO^{Obliv}}

\newcommand{\Dhop}{\hopDiameter{G}}
\newcommand{\sstar}{s^*}
\newcommand{\vD}{v_D}

\newcommand{\dist}{\operatorname{dist}}

\newcommand{\shortcutQuality}[1]{\mathrm{Shortcut}\-\mathrm{Quality}(#1)}
\newcommand{\hopDiameter}[1]{\mathrm{Hop}\-\mathrm{Diam}(#1)}

\newcommand{\todob}[1]{}   

\newcommand{\valSNC}{\tau}

\usepackage{enumitem}

\title{Undirected $(1+\eps)$-Shortest Paths via Minor-Aggregates: Near-Optimal Deterministic Parallel \& Distributed Algorithms}

\author{
Václav Rozhoň, ETH Zurich 
\and \circledR Christoph Grunau, ETH Zurich 
\and \circledR Bernhard Haeupler, ETH Zurich \& Carnegie Mellon University
\and \circledR Goran Zuzic, ETH Zurich
\and\circledR Jason Li, UC Berkeley
}

\titlenote{
VR and CG received funding from the European Research Council (ERC) under the European Unions Horizon
2020 research and innovation programme (grant agreement No. 853109).
BH was supported in part by NSF grants CCF-1814603, CCF-1910588, NSF CAREER award CCF-1750808, a Sloan Research Fellowship, funding from the European Research Council (ERC) under the European Union's Horizon 2020 research and innovation program (ERC grant agreement 949272), and the Swiss National Foundation (project grant 200021-184735)
GZ was supported in part by the Swiss National Foundation (project grant 200021-184735).
JL was supported by the Simons Foundation and the Simons Institute for the Theory of Computing.
}

\setcopyright{acmcopyright}
\acmPrice{15.00}
\acmDOI{10.1145/3519935.3520074}
\acmYear{2022}
\copyrightyear{2022}
\acmSubmissionID{stoc22main-p1051-p}
\acmISBN{978-1-4503-9264-8/22/06}
\acmConference[STOC '22]{Proceedings of the 54th Annual ACM SIGACT Symposium on Theory of Computing}{June 20--24, 2022}{Rome, Italy}
\acmBooktitle{Proceedings of the 54th Annual ACM SIGACT Symposium on Theory of Computing (STOC '22), June 20--24, 2022, Rome, Italy}

\begin{abstract}
\noindent This paper presents near-optimal deterministic parallel and distributed algorithms for computing $(1+\eps)$-approximate single-source shortest paths in any undirected weighted graph.
\medskip

\noindent On a high level, we deterministically reduce this and other shortest-path problems to $\tO(1)$ \footnote{We use $\tO$-notation to suppress polylogarthmic factors in the number of nodes $n$, e.g., $\tO(m) = m \log^{O(1)} n$. We use the term near-linear to mean $\tO(m)$.} Minor-Aggregations. 
A Minor-Aggregation computes an aggregate (e.g., max or sum) of node-values for every connected component of some subgraph.

\medskip

\noindent Our reduction immediately implies:
\begin{itemize}
    \item Optimal deterministic parallel (\pram) algorithms with $\tO(1)$ depth and near-linear work.
    \item Universally-optimal deterministic distributed (\congest) algorithms, whenever deterministic Minor-Aggregate algorithms exist. For example, an optimal $\tO(\hopDiameter{G})$-round deterministic \congest algorithm for excluded-minor networks.
\end{itemize}

\medskip

\noindent Several novel tools developed for the above results are interesting in their own right:

\begin{itemize}
    \item A local iterative approach for reducing shortest path computations ``up to distance $D$'' to computing low-diameter decompositions  ``up to distance $\frac{D}{2}$''. Compared to the recursive vertex-reduction approach of \cite{Li20}, our approach is simpler, suitable for distributed algorithms, and eliminates many derandomization barriers.
    \item A simple graph-based $\tO(1)$-competitive $\ell_1$-oblivious routing based on low-diameter decompositions that can be evaluated in near-linear work. The previous such routing~\cite{goranci2022universally} was $n^{o(1)}$-competitive and required $n^{o(1)}$ more work.   

    \item A deterministic algorithm to round any fractional single-source transshipment flow into an integral tree solution. 

    \item The first distributed algorithms for computing Eulerian orientations. 
\end{itemize}

\smallskip

\end{abstract}

\begin{document}
\maketitle
\section{Introduction}\label{sec:intro}

This paper gives essentially near-optimal\footnote{We refer to optimality up to polylogarithmic factors as \emph{near-optimality}, while \emph{almost-optimality} refers to optimality up to subpolynomial, i.e., $n^{o(1)}$, factors.} deterministic parallel algorithms for various $(1+\eps)$-approximate shortest-path-type problems in undirected weighted graphs. These problems include computing $(1+\eps)$-approximate shortest paths or shortest path trees from a single source (shortened to $(1+\eps)$-SSSP) and $(1 + \eps)$-approximate minimum transshipment (shortened to $(1+\eps)$-transshipment). Our algorithms run in polylogarithmic time and require near-linear work. This is optimal up to polylogarithmic factors. All prior parallel deterministic algorithms with polylogarithmic depth for $(1+\eps)$-SSSP required $n^{\Theta(1)}$ more work and, to our knowledge, no sub-quadratic work NC algorithm was known for $(1+\eps)$-transshipment. We also give deterministic distributed algorithms for the above problems, including the first non-trivial distributed deterministic $(1+\eps)$-transshipment algorithm. 

Computing a single-source shortest path is one of the most fundamental problems in combinatorial optimization. Already in 1956, Dijkstra \cite{dijkstra1959dijkstra} gave a simple deterministic $O(m \log n)$ time algorithm for computing an exact single-source shortest path tree in any directed graph. Parallel and distributed algorithms on the other hand have been subject of over three decades of extensive research \cite{ullman1991high,cohen1994polylog,klein1997randomized,spencer1997time,galil1997all,alon1997exponent,zwick1998all,shi1999time,Nan14,EN16,henzinger2016almost,BKKL17,Elk17,elkin2017linearsize,FN18,GL18,haeuplerLi2018sssp,elkin2019hopsets,EN19,henzinger2019deterministic,LPP19,CM20,elkinMatar2021detPRAMSSSP,goranci2022universally} with much remaining unknown. Much recent research has focused on $(1+\eps)$-approximate shortest paths in undirected graphs, the problem solved in this paper. We provide a detailed summary of this prior work next and describe related work on SSSP algorithms in other computational models or on other related problems, including all-pair-shortest-path and SSSP in directed graphs, in the full version of the paper. We focus on randomized and deterministic parallel algorithms with polylogarithmic depth, i.e., RNC and NC algorithms. 

\paragraph{Parallel Algorithm.}
The \pram model of parallel computation was introduced in 1979. After a decade of intense study, Karp and Ramachandran noted in their influential 1989 survey on parallel algorithms~\cite{karp1989survey} that the ``transitive closure bottleneck'' for shortest-path-type problem required work ``far in excess of the time required to solve the problem sequentially''. Indeed, these algorithms build on matrix multiplication or on transitive closure routines and require $O(n^{3})$ work. A line of work \cite{ullman1991high,klein1997randomized,spencer1997time,cohen1997using,shi1999time} provided algorithms with different work-time trade-offs but without much improvement on the cubic work bound for fast parallel algorithms.



In a major breakthrough, Cohen~\cite{cohen1994polylog} gave a randomized $(1 + \eps)$-SSSP algorithm based on hopsets with $\tO_\rho(1)$ time and $O(m^{1+\rho})$ work for any constant $\rho > 0$. Algorithms with slightly improved time-work tradoffs were given~\cite{elkin2017linearsize,elkin2019hopsets} but these algorithms still required a polynomial $m^{\rho}$ factor more work than Dijkstra's algorithm to parallelize shortest path computations to polylogarithmic parallel time. 

This remained the state-of-the-art until 2020 when methods from continuous optimization developed by Sherman~\cite{She13,She17b} enabled further progress on shortest-path type problems. In particular, Sherman gave an $m \cdot n^{o(1)}$ sequential algorithm for the $(1+\eps)$-approximate minimum transshipment problem. \emph{Transshipment}, also known as uncapacitated min-cost flow, Wasserstein distance, or optimal transport, is the other major problem solved in this paper. The input to the transshipment problem on some weighted graph $G$ consists of a \emph{demand} $b$ which specifies for each node a positive or negative demand value with these values summing up to zero. The goal is to find a flow of minimum cost which sends flow from nodes with surplus to nodes with positive demands. Here the cost of routing a unit of flow over an edge is proportional to the weight/length of the edge. For example, an optimal transshipment flow for a demand with $b(s)=-1$ and $b(t)=1$ is simply a unit flow from $s$ to $t$ which decomposes into a distribution over shortest $(s,t)$-paths. We remark that transshipment is in many ways a powerful generalization of shortest-path problems, including computing SSSP (trees), with the exception that extracting trees or paths from a continuous transshipment flow often remains a hard problem itself which requires highly nontrivial rounding algorithms. While having an innocent feel to them, such rounding steps have a history of being inherently randomized and being algorithmic bottlenecks for shortest-path problems on several occasions. 

Sherman's almost-linear time transshipment algorithm did not imply anything new for SSSP at first, but it inspired two independent approaches by \cite{Li20} and Andoni, Stein, and Zhong \cite{AndoniSZ20} that led to near-optimal randomized parallel shortest path algorithms with polylogarithmic time and near-linear work. In particular, Li~\cite{Li20} improved several parts of Sherman's algorithm and combined it with the vertex-reduction framework that Peng~\cite{Pen16} had introduced to give the first near-linear time algorithm for $(1+\eps)$-maximum-flow. \cite{Li20} also adopts a randomized rounding algorithm of \cite{BKKL17} based on random walks to extract approximate shortest-path trees from transshipment flows. Overall, this results in randomized PRAM algorithms for both $(1+\eps)$-transshipment and $(1+\eps)$-SSSP with polylogarithmic time and near-linear work.
Concurrently, Andoni, Stein, and Zhong \cite{AndoniSZ20} achieved the same result for $(1+\eps)$-transshipment by combining Sherman's framework with ideas from hop-sets. They also provide a randomized rounding based on random walks algorithm that can extract an $s$-$t$ shortest path. Later, this rounding was extended in Zhong's thesis~\cite{zhong2021new} to extract the full $(1+\eps)$-SSSP tree. 

It is remarkable that all parallel $(1+\eps)$-SSSP algorithms described above (with the exception of the super-quadratic work algorithms before Cohen's 25 year old breakthrough) crucially rely on randomization for efficiency. Indeed, the only modern deterministic parallel $(1+\eps)$-SSSP algorithm is a very recent derandomization of Cohen's algorithm by Elkin and Matar~\cite{elkinMatar2021detPRAMSSSP}. This algorithm suffers from the familiar $O(m^{1+\rho})$ work bound for a polylogarithmic time parallel algorithm. For the $(1+\eps)$-transshipment problem even less is known. Indeed we are not aware of any efficient deterministic parallel algorithm before this work and would expect that much larger polynomial work bounds would be required to achieve a deterministic algorithm with polylogarithmic parallel time.

In this paper, we give deterministic parallel algorithms for both the $(1 + \eps)$-SSSP and the $(1+\eps)$-transshipment problem in undirected weighted graphs. Both algorithms only require near-linear work and run in polylogarithmic time. This solves these two and various other shortest-path-type problems optimally, up to polylogarithmic factors. 

\begin{restatable}[Deterministic Parallel SSSP and Transshipment]{theorem}{mainparallel}
\label{thm:main_parallel}
There is a deterministic parallel algorithm that, given an undirected graph with nonnegative weights, computes a $(1+\eps)$-approximate single-source shortest path tree or a $(1+\eps)$-approximation to minimum transshipment in $\tO(m \cdot \eps^{-2})$ work and $\tO(1)$ time in the \pram model for any $\eps \in (0,1]$.
\end{restatable}

\paragraph{Distributed Algorithm.}
We also give new distributed $(1+\eps)$-SSSP and $(1+\eps)$-transshipment algorithms in the standard \congest model of distributed computing. 
For distributed algorithms a lower bound of $\Omega(\sqrt{n} + \Dhop)$ rounds is known for worst-case topologies~\cite{elkin2006unconditional,sarma2012distributed}, where $\Dhop$ is the unweighted (i.e., hop) diameter of the network. We refer to this as the \emph{existential lower bound} since it depends on the parameterization by $(n, \Dhop)$, as opposed to the later-discussed \emph{universal lower bound} which does not.

In his 2004 survey on distributed approximation, Elkin~\cite{elkin2004distributed} pointed to the distributed complexity of shortest paths approximations as one of two fundamental and wide-open problems in the area---as no non-trivial algorithms complementing the above lower bounds were known.

Since then, the $(1+\eps)$-SSSP problem is one of the most studied problems. We only give a brief summary here: Lenzen and Patt-Shamir gave an $\tO(n^{1/2 + \eps} + \Dhop)$-round $O(1 + \frac{\log 1/\eps}{\eps})$-SSSP algorithm \cite{lenzen2013fast} and a $(1+o(1))$-approximation in $\tO(n^{1/2}\cdot \Dhop^{1/4}+\Dhop)$ rounds was given by Nanongkai \cite{Nan14}. 
Building on Sherman's transshipment framework mentioned before, \cite{BKKL17} gave a randomized $\tO(\eps^{-3} \cdot (\sqrt n + \Dhop))$ algorithm.
The first algorithm improving over the $\Omega({\sqrt{n}})$ barrier was given in \cite{haeuplerLi2018sssp} with a running time of $\shortcutQuality{G} \cdot n^{o(1)}$ albeit with a bad $n^{o(1)}$-approximation. 
This was recently improved by \cite{goranci2022universally} to a $(1+\eps)$-approximation with the same round complexity. 
Using \cite{ghaffari_haeupler2021shortcuts_in_minor_closed}, this gives a $(1+\eps)$-approximation algorithm with round complexity $\Dhop \cdot n^{o(1)}$ for any excluded minor topology (e.g., planar graphs). 
All the above algorithms are randomized: The only deterministic distributed algorithm is the $(1 + o(1))$-SSSP algorithm of \cite{henzinger2019deterministic} with a $(n^{1/2+o(1)} + \Dhop^{1+o(1)})$ running time. For $(1+\eps)$-transshipment, the only known sublinear distributed algorithm is the randomized $\shortcutQuality{G} \cdot n^{o(1)}$-round algorithm of \cite{goranci2022universally}. No non-trivial  deterministic distributed $(1+\eps)$-transshipment algorithm (with sub-linear round complexity) was known prior to this work.

The distributed results of this paper include deterministic distributed algorithms for $(1+\eps)$-SSSP, its generalization $(1+\eps)$-\setssp (in which we are looking for distances from a subset of nodes), and  $(1+\eps)$-transshipment. Our results improve over the previous best running times of deterministic distributed algorithms for each of the above problems, both in the worst-case and for any excluded minor graph. Both the round and the message complexities (in $KT_0$, see \cite{peleg2000distributed} for the definition) of our algorithms are (existentially) optimal, up to $\poly\log n$ factors, in all cases. 


\begin{restatable}[Deterministic Distributed SSSP and Transshipment]{theorem}{maindistributed}
\label{thm:main_distributed}
There are deterministic \congest algorithms that, given an undirected graph with non-negative weights, compute a $(1+\eps)$-approximate \setssp forest or a $(1+\eps)$-transshipment solution for any $\eps \in (0,1]$. Our algorithms have an optimal message complexity of $\tO (m) \cdot  \eps^{-2}$ and are guaranteed to terminate 
\begin{enumerate}
    \item within at most $\tO(\sqrt{n} + \Dhop) \cdot \eps^{-2}$ rounds and
    \item within at most $\tO(\Dhop) \cdot  \eps^{-2}$ rounds if $G$ does not contain any $\tO(1)$-dense minor.
\end{enumerate}
\end{restatable}


\underline{Universal optimality.} Recently, the pervasive $\tilde{\Omega}(\Dhop + \sqrt{n})$ distributed lower bounds have been extended to a near-tight universal lower bound~\cite{haeupler2021universally} which shows that most optimization problems including all problems studied in this paper require $\tilde{\Omega}(\shortcutQuality{G})$ rounds on any communication graph $G$. Here $\shortcutQuality{G}$ is a natural graph parameter (we refer the interested reader to \cite{haeupler2021universally} for a formal definition). For experts interested in universally-optimal distributed algorithms we remark that the results of this paper imply strong conditional results. We state these results for the interested reader in \Cref{sec:main_theorems}.


\subsection{Technical Overview}

\begin{figure}
    \centering
    \includegraphics[width = .5\textwidth]{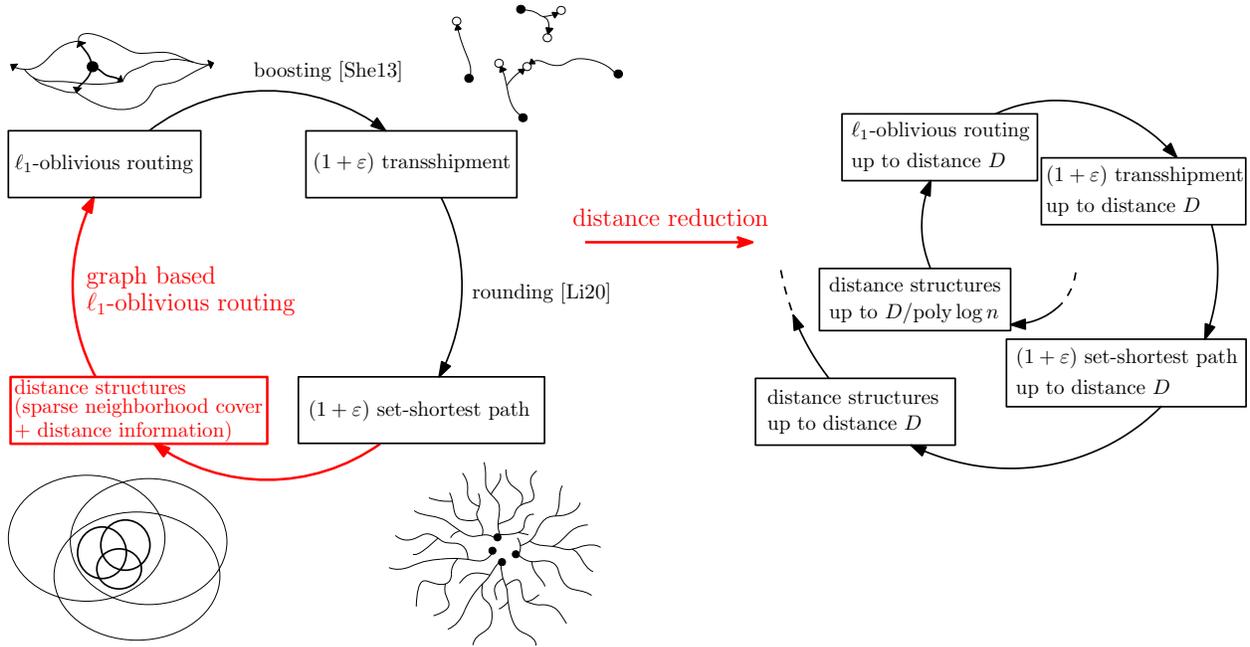}
    \caption{This figure summarizes our approach with our contributions marked in red. It is known that a $\tO(1)$-competitive solution to $\ell_1$-oblivious routing can be boosted to a $(1+\eps)$-approximate transshipment which can then be rounded to yield a $(1+\eps)$-approximate SSSP tree. We close the circle of reductions by constructing certain distance structures from the approximate SSSP tree. Think of those distance structures as a family of clusterings on different scales, storing some additional distance information. One of our main technical contributions is the efficient construction of a $\tO(1)$-approximate oblivious routing from the distance structures, which closes the loop. \\
    This loop of reductions itself is not very useful as the complexity of the problems does not decrease. 
    We break this loop by our \emph{distance reduction} framework by showing one can construct distance structures ``up to distance $D$'' using distance structures ``up to distance $D /\poly\log n$. 
    After $O(\log n)$ iterations, we build up the desired solution to any of the four problems, in particular to approximate transshipment and set-SSSP. }
    \label{fig:main}
\end{figure}

While the results for SSSP and transshipment are important, the main impact of this paper will likely be the new tools and algorithmic ideas developed in this paper. We expect these ideas to be applicable beyond shortest path problems. Next, we give a brief summary of the most relevant parts and ideas of prior work which are needed to understand our work and put in proper context. We then give an informal high-level overview of our new algorithm and some of the new tools developed for it. A readable and more precise technical proof overview is given in \Cref{sec:localiterative}.

\subsubsection{Background and Prior Work}

\paragraph{ Transshipment Boosting, $\ell_1$-Oblivious Routing, and Transshipment Flow Rounding} 
All modern SSSP-algorithms, including ours, compute shortest paths via the transshipment problem (see text before \Cref{thm:main_parallel} for a definition). The key idea in this approach of Sherman~\cite{She13, She17b} is that even a rather bad $\alpha$-approximation algorithm for (dual) transshipment can be \emph{boosted} to a $(1+\eps)$-approximation. This is achieved via the multiplicative weights method (or equivalently: gradient descent) and requires only $\poly($ $\alpha, \eps^{-1}, \log n)$ invocations of the $\alpha$-approximation algorithm~\cite{She17b,BKKL17,zuzic2021simple}. A particularly convenient way of obtaining such a boostable $\alpha$-approximation is to design a linear matrix $R$ which maps any node-demand $b$ to an $\alpha$-approximate transshipment flow for $b$. Such a matrix is called an $\ell_1$-oblivious routing because linearity forces each node to route its demand obliviously, i.e., without knowledge of the demand on other nodes. As mentioned before, in order to obtain an actual $(1+\eps)$-SSSP tree from a transshipment flow most algorithms~\cite{Li20} are using an approach of \cite{becker2019low} which produces a $(1+\eps)$-SSSP tree after $O(\log n)$ adaptive applications of some rounding algorithm. All parts of this reduction except for the rounding algorithm are deterministic.

Putting all these pieces together proves that all one needs to obtain an efficient (deterministic) algorithm for both the  $(1+\eps)$-transshipment and the $(1+\eps)$-SSSP problem is a $\tO(1)$-competitive $\ell_1$-oblivious routing that can be evaluated efficiently and a (deterministic) rounding algorithm. Our algorithm uses these steps in a black-box fashion except that we need to replace the randomized rounding algorithm by a new deterministic transshipment flow rounding procedure.

\paragraph{ Vertex Reduction Framework}
Unfortunately, it is clear that $\ell_1$-oblivious routing cannot be done without having some approximate shortest path information. This chicken and egg problem is resolved via a clever vertex-reduction framework of Li~\cite{Li20}.  
The vertex-reduction framework relies on a cyclic sequence of efficient problem reductions. For transshipment~\cite{Li20}, these problems are SSSP, transshipment, $\ell_1$-oblivious routing, and $\ell_1$-embedding. Specifically:
\begin{itemize}
\item To solve $(1+\eps)$-SSSP on $G$, it is sufficient to solve $(1+\eps)$-transshipment on $G$, 
\item for which it is sufficient to construct $\tilde{O}(1)$-competitive $\ell_1$-oblivious routing on $G$ (via boosting),
\item for which it is sufficient to construct an $\tilde{O}(1)$-distortion $\ell_1$-embedding on $G$,
\item for which it is sufficient to solve $O(\log^2 n)$ instances of $\tilde{O}(1)$-SSSP on a sequence of graphs $G'_1, G'_2, \ldots, G'_{O(\log^2 n)}$, which are resolved recursively.
\end{itemize}
As stated, this simply reduced $(1+\eps)$-SSSP on $G$ to multiple $\tilde{O}(1)$-SSSP on graphs that are slightly larger than $G$, which isn't particularly helpful on its own. The key idea to transform this ``branching cycle'' into a branching spiral that terminates (fast) is to add a step into the cycle which applies ultrasparsification on $G'_i$, which is an operation that transforms a graph with $n$ nodes to a smaller graph with $\frac{n}{\gamma}$ nodes such that distances in the smaller graph $\tilde{O}(\gamma)$-approximate distances in the larger graph. Applying this step, we can reduce $(1 + \eps)$-SSSP on a graph with $n$ nodes to $O(\log^2 n)$ instances of $\tilde{O}(1)$-SSSPs on smaller graphs with $O\left(\frac{n}{\log^{100} n}\right)$ nodes. It is easy to see that the total size of all recursive calls falls exponentially on each subsequent level, hence the total parallel runtime is $\tilde{O}(1)$ and the total work $\tilde{O}(m)$ (since a single sequence of reductions requires $\tilde{O}(m)$ work).

\paragraph{ Minor Aggregates and the Low-Congestion Shortcuts Framework}
The low-congestion shortcuts framework was originally intended for designing simple and efficient \emph{distributed} graph algorithms and was developed over a long sequence of works~\cite{GH16,haeupler2016low,haeupler2016near,haeupler2018round,haeupler2018minor,ghaffari_haeupler2021shortcuts_in_minor_closed,haeupler2021universally,kogan2021low,ghaffari2021hop,GHR21}. However, this paper argues that the framework provides a natural language even for developing fast parallel algorithms. Indeed, we present our parallel SSSP algorithm using the framework and our hope is that this choice simplifies the exposition, even before considering the benefits of immediately obtaining distributed results.

We describe our algorithms in the recently introduced Minor-Aggregation model~\cite{goranci2022universally, ghaffari2021universally}, which offers an intuitive interface to the recent advancements in the low-congestion shortcut framework. In the \aggregate model, one can (1) contract edges, thereby operating a minor, (2) each (super)node in the contracted graph can compute an aggregate (e.g., min, max, sum) of surrounding neighbors, and (3) add $\tilde{O}(1)$ arbitrarily-connected virtual nodes over the course of the algorithm. The goal is to design $\tO(1)$-round algorithms in this model. Such a \aggregate algorithm can be compiled to a near-optimal algorithm in both the parallel and distributed settings.


\subsubsection{Our New Tools}


\paragraph{ Near-Optimal Graph-Based $\ell_1$-oblivious routing}
The first key contribution of this paper is a new construction of graph-based $\ell_1$-oblivious routing with drastically improved guarantees from the so-called sparse neighborhood covers. A sparse neighborhood cover of distance scale $D$ is a collection of $O(\log n)$ clusterings (partitioning of the node set into disjoint \emph{clusters}) such that (1) each cluster has diameter at most $D$, and (2) each ball in $G$ of radius $D / (\log^{C} n)$ is fully contained in at least one cluster (for some fixed constant $C > 0$). Specifically, given sparse neighborhood covers for all $O(\log n)$ exponentially-increasing distance scales $\beta, \beta^2, \beta^3, \dots, \poly(n)$ for some $\beta = \tO(1)$ along with some extra distance information, we construct an $\tO(1)$-competitive $\ell_1$-oblivious routing. 
The algorithm greatly differs from all prior approaches, as all of them had inherent barriers preventing them from achieving deterministic near-optimality, which we describe below.

\underline{Derandomization issues:} Efficient constructions of an $\ell_1$-oblivious routing either use an $\ell_1$-embedding or so-called low-diameter decompositions; this is a clustering problem whose deterministic version is also known as a sparse neighborhood cover. 
Sparse neighborhood covers were introduced in the seminal work of \cite{awerbuch1990sparse} and applied with great success for many problems, including approximate shortest paths and other distance based problems \cite{cohen98pairwise_covers,bartal2021RamseyEmbeddings}. 

Since an efficient deterministic $\ell_1$-embedding is not known, we need a deterministic construction of sparse neighborhood covers.  
Luckily, many of the ideas required to derandomize the computation of sparse neighborhood covers were developed very recently by a sequence of papers that derandomized the closely related network decomposition problem which is, essentially, an unweighted version of the sparse cover problem~\cite{rozhon_ghaffari2019decomposition,ghaffari_grunau_rozhon2020improved_network_decomposition,chang_ghaffari2021strong_diameter}.
In \cite{elkin2022decompositions}, these unweighted results were extended to an algorithm constructing sparse neighborhood covers for weighted graphs and this is the result that we use to solve the shortest path problem here. 

\paragraph{ Distance-reduction framework: Iterative and Locality-Friendly}
As we explained earlier, our new $\ell_1$-oblivious routing construction reduces the SSSP problem in $G$ to computing sparse neighborhood covers along with some extra distance information.
Clearly, this requires our SSSP computation to compute some distance information, leading again to a chicken and egg problem. Unfortunately, directly applying the vertex-reduction framework is not compatible with our desire to obtain fast minor-aggregation (or distributed) algorithms. In particular, at the lowest level of the recursion, the vertex-reduction framework generates a polynomial number of constant-size SSSP problems, each of which essentially corresponds to a minor of $G$. While these problems can each be solved in constant time in the parallel setting, the fact that these minors can be arbitrarily overlapping and each correspond to large diameter subsets in $G$ means that polynomially many instead of the desired polylogarithmic number of minor aggregations are necessary.

This paper therefore designs a novel, completely different, and more locality-friendly complexity-reduction framework: \emph{distance reduction}. On a high-level, we show that obtaining ``sparse neighborhood covers up to distance scale $D$'' can be reduced to several ``shortest path computations up to distance $D$'' which are computable from ``sparse neighborhood covers up to distance $D/\poly\log n$'' (see \cref{fig:main}). 
This iterative and local nature of our distance-reduction framework directly translates into a small overall number of minor aggregations; this is in stark contrast to the inherently recursive vertex-reduction framework (which recurses on different graphs, requiring a recursive approach).
As a nice little extra, combining our distance-reduction framework with the new $\ell_1$-embedding makes for an algorithm that is (in our not exactly unbiased opinion) a good bit simpler than the previous algorithms of \cite{Li20,AndoniSZ20}. 
We note that this framework is more reminiscent of older approaches to hopset constructions~\cite{cohen1994polylog, EN16}. In these approaches, hopsets over longer paths which use at most $\ell / 2$ edges are used to bootstrap the construction of hopsets over paths using at most $\ell$ edges.


A key definition to formalize what exactly the ``up to distance $D$'' in our distance-reduction framework means is the following. We attach a virtual, so-called, \emph{cheating node} and connect it to all other nodes with an edge of length $D$. This new graph naturally preserves distance information ``up to distance $D$'' and the complexity of computing distances increases as $D$ increases.










\paragraph{ Derandomization: Deterministic Transshipment Flow Rounding via Eulerian Tours}
All previous transshipment approaches crucially use randomization. The most significant challenge we had to overcome in making our results deterministic stem from the following issue: The only problem that remains to be derandomized in this paper is \emph{rounding} a (fractional) transshipment solution to (a flow supported on) a tree, a crucial step in the distance-reduction framework. We prove the following theorem (only the parallel version is stated for simplicity). 


\begin{theorem}\label{thm:eulerrounding-intro}
There is a deterministic parallel algorithm which takes as an input any fractional transshipment flow $f$ satisfying some single-source transshipment demand $b$ in the graph $G$ and in near-linear work and polylogarithmic time outputs a flow $f'$ of equal or smaller cost which is supported on some tree in $G$.
\end{theorem}

While such a flow rounding seems an unlikely bottleneck for a deterministic SSSP algorithm, we remark that even for randomized algorithms a lot of complexity has come from this rounding step in the algorithms of \cite{Li20, AndoniSZ20, BKKL17}---all these approaches are based on random walks and are inherently randomized.


Note that if the fractional flow $f$ is acyclic (has no \emph{directed} cycles), then there is a trivial randomized rounding which simply samples one outgoing edge for each node in $G$ through which flow is routed, choosing the probability of each edge proportional to the outflow in $f$. It is easy to see that retaining the edges that are in the connected component of the source truthfully samples a tree-supported flow from $f$. The complexity comes in once $f$ is not acyclic as the sampled edges can now create many connected components. This requires finding the cycles in these components, contracting the edges and again running a randomized out-edge sampling on the remaining graph. 

Deterministically none of the above works. Indeed, we are not aware of any simple(r) way of deterministically obtaining a tree-supported flow even if $f$ is acyclic. Our rounding procedure can be seen as a generalization of an algorithm of Cohen's rounding~\cite{cohen1995approximateMaxFlow}, which can be used to round any fractional transshipment flow to an integral flow by scaling flow values to only leave integral and half-integral flow values and then finding an Eulerian tour covering all half-integral edges. Pushing one half-unit of flow in the cheaper direction of this Eulerian tour makes all flow values fully integral and allows the scaling to be reduced. At the end of this procedure, all flow values are integral but this does not guarantee that the flow is supported on a tree. To eliminate any non-tree like parts of the flow we show how to keep the algorithm running and find further Eulerian tours until the flow becomes tree-supported. 

\paragraph{ A Distributed Eulerian Tour Algorithm}
The problem of computing Eulerian Tours can be stated as follows. Given an undirected graph with all degrees even, direct the edges in such a way that the out-degree of every node equals the in-degree. Note that we do not require connectedness. While computing Eulerian tours is a well-known parallel primitive which can be efficiently computed in near-linear work and polylogarithmic time \cite{atallah_vishkin1984euler_pram}, there are, to our knowledge, no distributed algorithms known for this problem. 

Therefore, to also give \emph{distributed} SSSP algorithms in this paper we need to design efficient distributed Eulerian-tour algorithms that can then be used in the Eulerian-tour-based rounding procedure of \Cref{thm:eulerrounding-intro}.  We build the first algorithm computing such an Eulerian tour orientation by using algorithms from \cite{parter_yogev2019cycle_covers_minor_closed,parter_yogev2019cycle_decomp_near_linear} for low-congestion cycle covers. Interestingly, these low-congestion cycle covers were only developed for the completely unrelated purpose of making distributed computation resilient to Byzantine faults introduced by an adversary in a recent line of work \cite{parter_yogev2019cycle_covers_minor_closed,parter_yogev2019cycle_decomp_near_linear,hitron_parter2021adversarial_compilers,hitron2021broadcast}. 

\begin{theorem}[Informal]\label{thm:euler-intro}
There is a deterministic \congest algorithm which, given any Eulerian subgraph $H$ of the network $G$ as an input, computes an Eulerian tour orientation in $\tO(\sqrt{n} + \Dhop)$ rounds or $\tO(\Dhop)$ rounds if $G$ is an excluded-minor graph. 
\end{theorem}

The most general and fully formal statements of all results proven in this paper are given in \Cref{sec:main_theorems}.

\section{SSSP via Minor-Aggregations: A Local Iterative Reduction Cycle}
\label{sec:localiterative}

In \cref{sec:intro} and \cref{fig:main}, we gave an idealized and informal description of our algorithm.
The real set of reductions our algorithm is built on is not quite the perfect cycle from \cref{fig:main}, but instead it looks like \cref{fig:oracles}. 
In \cref{sec:key_definitions}, we give a formal definition for each part of this new ``cycle'' and explain how the parts of the new cycle correspond to parts in the old cycle. 
Once we have defined each part of the new cycle, each arrow in \cref{fig:oracles} corresponds to a formal statement.
These formal statements can be found in \cref{sec:implications_cycle}, along with informal explanations how the statements can be proven. \\
Finally, \cref{sec:main_theorems} contains statements of our main theorems together with simple proof sketches.
The formal proofs for all the results stated in \cref{sec:implications_cycle} and \cref{sec:main_theorems} can be found in the full version of the paper.

\begin{figure}[ht]
    \centering
    \includegraphics[width = .5\textwidth]{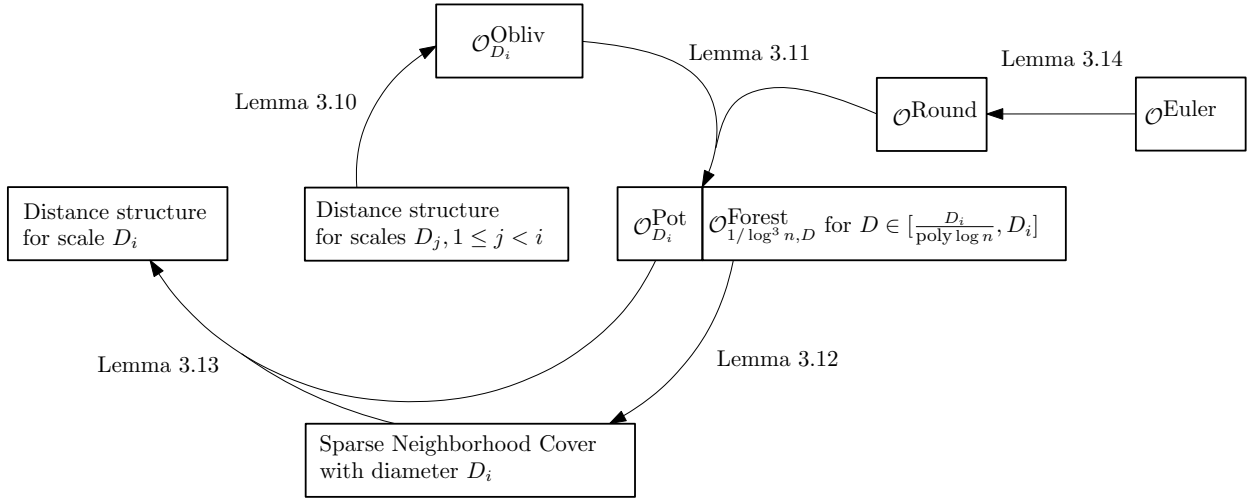}
    \caption{The figure illustrates a single iteration of our local iterative reduction cycle. At the beginning of the cycle, the algorithm has already computed a distance structure for every scale $D_j$ with $D_j < D_i$. After completion of the cycle, the algorithm has computed a distance structure for scale $D_i$, under the assumption that it has access to the rounding oracle $\oRound$ (or the \euler oracle $\oEuler$). \cref{lem:cycleone} states this result formally.
    Moreover, each arrow in the cycle corresponds to a formal statement of the following form: Given $A$ (and $B$), then one can efficiently (in $\tilde{O}(1)$ \aggregate rounds) compute $C$. }
    \label{fig:oracles}
\end{figure}

\subsection{Key Definitions and Oracles}
\label{sec:key_definitions}

The definition of the oblivious routing oracle $\oObliv_{D_i}$ relies on the definition of the graph $G_{S,D}$ and the definition of distance scales.

We start with the definition of the graph $G_{S,D}$.

\begin{definition}[Graph $G_{S,D}$]
Let $S \subseteq V$ and $D \in [\poly(n)]$. We construct the graph $G_{S,D}=(V \cup \{\vD,\sstar\}, E \cup \{\{\vD,u\} \mid u \in V\} \cup \{\{\sstar,u\} \mid u \in S\})$ by adding two additional nodes $\vD$ and $\sstar$ to $G$. The node $\vD$ is connected to each node in $V$ and the node $\sstar$ is connected to all the nodes in $S$, with all the new edges having a weight of $D$. Moreover, we denote with $G_S$ the unweighted graph one obtains by discarding the edge weights in the weighted graph $G_{S,D}$.
\end{definition}

Informally speaking, $G_{S,D}$ preserves distance information only up to distance $2D$. More formally, any shortest path between two nodes in $G$ remains a shortest path in $G_{S,D}$ if this path is of length at most $2D$. Shortest paths in $G$ longer than $2D$, on the other hand, are not preserved in $G_{S,D}$, since any two nodes in $G_{S,D}$ have a shortest path of length $2D$ via $\vD$.
Moreover, consider a shortest path tree from vertex $\sstar$ up to distance $2D$. If one removes $\sstar$ and $\vD$ from this tree, then the resulting forest is a shortest path forest in $G$ from set $S$ up to distance $D$. 

We next give the definition of distance scales, along with defining global parameters which are used throughout the paper.
\begin{definition}[Distance Scales]\label{def:distance_scales}
We set $\valSNC = \log^7(n)$, $\beta = 8\valSNC$ and $D_i = \beta^i$.

A distance scale is a value contained in the set $\{D_i \colon i \in \mathbb{N}, D_i \leq n^2 \max_{e \in E} \ell(e)\}$.
In particular, if we denote with $i_{max}$ the largest integer $i$ for which $D_{i}$ is a distance scale, then $D_{i_{max}} \geq \diam(G)$.
\end{definition}

We are finally ready to formally define oracle $\oObliv_{D_i}$, the $\ell_1$-Oblivious Routing Oracle for scale $D_i$. It corresponds to computing an ``$\ell_1$-oblivious routing up to distance $D_i$'' in \cref{fig:main}.

\begin{restatable}[$\ell_1$-Oblivious Routing Oracle for scale $D_i$ --- $\oObliv_{D_i}$]{definition}{oblivousroutingoracle}
\label{def:oracle_oblivousrouting}
The $\ell_1$-Oblivious Routing Oracle for scale $D_i$, $\oObliv_{D_i}$, takes as input a set $S \subseteq V$ as well as a demand $b$ and a flow $f$ for $G_{S,D_i}$. It outputs $R_S b$ and $R_S^T f$ for a fixed $\tilde{O}(1)$-competitive $\ell_1$-oblivious routing $R_S$ for $G_{S,D_i}$.
\end{restatable}

Next, we give the formal definition of the forest oracle $\oDist_{\eps,D}$.

\begin{restatable}[$(1+\eps)$-Approximate Forest for distance $D$ rooted at $S$ / Forest Oracle --- $\oDist_{\eps,D}$]{definition}{distanceoracle}\label{def:oracle_dist}
The forest oracle $\oDist_{\eps,D}$ takes as input a node set $S \subseteq V$.
The output is a $(1+\eps)$-approximate forest for distance $D$ rooted at $S$, which is a forest $F$ rooted at $S$ such that the following holds.
\begin{enumerate}
    \item For every $u \in V(G)$ with $\dist_G(S,u) \le D$ we have $u \in V(F)$.
    \item For every $u \in V(F)$, $\dist_F(S,u) \le (1+\eps)D$. 
\end{enumerate}
\end{restatable}

For a given distance scale $D_i$, having access to $\oDist_{\frac{1}{\log^3(n)},D}$ for every $D \in \left[ \frac{D_i}{\valSNC}, D_i\right]$ corresponds to ``$(1+\eps)$ set-shortest path up to distance $D_i$'' in \cref{fig:main}.

The notion of ``$(1+\eps)$ transshipment up to distance $D$'' in \cref{fig:main} does not have a direct correspondence in \cref{fig:oracles}.
However, the potential oracle $\oPot_{D_i}$ outputs a potential capturing a certain version of (dual) transshipment potentials. 
The oracle $\oPot_{D_i}$ is defined as follows:

\begin{definition}[Potential for scale $D_i$ with respect to set $S$ / Potential Oracle --- $\oPot_{D_i}$] \label{def:potential}
The potential oracle $\oPot_{D_i}$ takes as input a node set $S \subseteq V$.
The output is a potential $\phi_{S, D_i} \in \R^{V(G_{S, D_i})}$ for scale $D_i$ with respect to $S$.
A potential for scale $D_i$ with respect to a set $S \subseteq V$ is a non-negative potential $\phi_{S,D_i}$ such that:
\begin{enumerate}
    \item $\forall v \in S:\ \phi_{S,D_i}(v) = 0$ and
    \item $\dist_G(v,S) \geq \frac{D_i}{\valSNC}$ implies $\phi_{S,D_i}(v) \geq 0.5\frac{D_i}{\valSNC}$.
\end{enumerate} 
\end{definition}

We next give the definition of a distance structure for scale $D_i$.
\begin{definition}[Distance Structure for scale $D_i$]
\label{def:distance_structure}
A distance-structure for scale $D_i$
consists of a sparse neighborhood cover with covering radius $\frac{D_i}{\valSNC}$. Moreover, each cluster $C$ in one of the clusterings comes with 
\begin{enumerate}
    \item a tree $T_C$ of diameter at most $D_i$ which spans $C$ and is rooted at some node $v_C \in C$. We refer to $v_C$ as the cluster center.
    \item A potential for scale $D_i$ with respect to $V \setminus C$ (known to nodes in $C$). \alert{we never defined what potentials with respect to a subset means.. also, with respect to which weighted graph is this? (I suppose $G_{D_i}$ or something.. but this is important to say?}
\end{enumerate}
\end{definition}

\alert{also, I think that the dist structure potentials are quite hard to parse for non-authors of this paper .. we should explain that this potentials with respect to the outside of the cluster actually means we want to know the distance to the boundary, but we are happy with an approximate distance that satisfies the potential smoothness properties or whatever.}

In \cref{fig:main}, ``distance structures up to distance $D$'' correspond to having a distance structure for every scale $D_j$ with $D_j \leq D$.

In \cref{sec:intro} we outlined the problem of rounding a transshipment flow, and that rounding can be deterministically reduced to solving $\tilde{O}(1)$ \euler problems. 
We now define the oracles for the two problems.

The rounding oracle is defined as follows. \alert{in sec 5 we actually give an implementation for a perfect $\eps = 0$ rounding oracle .. Christoph explained that this is to future-proof the proof as someone might come up with a (distributed) rounding faster algorithm in the future. However, without any explanation here this is extremely confusing (at least to me). Can we please add a explanation why we state it in this form.}

\begin{restatable}[Rounding Oracle --- $\oRound_{\eps}$]{definition}{roundingoracle}
\label{def:oracle_rounding}
The Rounding Oracle $\oRound_{\eps}$ takes as input a weighted graph $H$ with length function $\ell_H$ and a flow $f$ on $H$. The weighted graph $H$ needs to be a subgraph of some graph $H'$ that one can obtain from $G$ by adding up to $\tilde{O}(1)$ virtual nodes and adding edges of arbitrary nonnegative length incident to the virtual nodes. The flow $f$ needs to satisfy the following condition. Let $b$ be the demand that $f$ routes. Then, $b(v) \geq 0$ for all $v \in V(H)$ except for some $s \in V(H)$ called the source. 
The output is a rooted tree $T$ with root $s$ spanning all vertices of $H$ with non-zero demand such that $\sum_{v:b(v) \neq 0}b(v)\dist_T(s,v) \le (1+\eps')\ell_H(f)$ with $\eps' = \min \left( \eps,\frac{1}{20\log^3(n)\valSNC} \right)$. 
\end{restatable} 

Note that $\oRound_{\eps}$ takes as input a weighted graph that can have more vertices than $G$. This allows us to use $\oRound_{\eps}$ to round a transshipment flow defined on the graph $G_{S,D}$. 
Moreover, we sometimes just write $\oRound$, without a precision parameter $\eps$, which we define as $\oRound = \oRound_1 = \oRound_{\frac{1}{20\log^3(n)\valSNC}}$.

The \euler oracle is defined next. 

\begin{restatable}[\euler Oracle --- $\oEuler$]{definition}{euleroracle}
\label{def:oracle_euler}
The \euler oracle $\oEuler$ takes as input a Eulerian graph $H$. The graph $H$ has to be a subgraph of some graph $H'$ that one can obtain from $G$ by adding up to $\tilde{O}(1)$ virtual nodes and adding any edges incident to the virtual nodes.
The output is an orientation of the edges of $H$ such that the in-degree of every node is equal to its out-degree.
\end{restatable}

\subsection{Formal Statements Corresponding to \cref{fig:oracles}}
\label{sec:implications_cycle}

For the sake of this section, we say that we can solve a problem or compute a structure \emph{efficiently} if there exists a \aggregate algorithm for the task that runs in $\tilde{O}(1)$ rounds. 

In this subsection, we give one formal statement for each arrow in \cref{fig:oracles}, as promised at the beginning of \cref{sec:localiterative}.
Before that, we state a result which captures the main essence of our local iterative reduction cycle. 

\begin{restatable}[Main Lemma]{lemma}{cycleone}
\label{lem:cycleone}
Assume a distance structure for every scale $D_j$ smaller than scale $D_i$ and oracle $\oRound$ are given. A distance structure for scale $D_i$ can be efficiently computed.
\end{restatable}

\cref{lem:cycleone} is a simple corollary of the next four lemmas. The formal proof can be found in the full version of the paper.

\begin{restatable}{lemma}{cycletwo}
\label{lem:cycletwo}
Assume a distance structure for every scale $D_j$ smaller than scale $D_i$ is given. Then, $\oObliv_{D_i}$ can be efficiently computed. 
\end{restatable}
\cref{lem:cycletwo} follows from a simple adaption of our $\tO(1)$-competitive $\ell_1$-oblivious routing construction, explained in the full version of the paper.

\begin{restatable}{lemma}{cyclethree}
\label{lem:cyclethree}
Assume oracle $\oRound$ and $\oObliv_{D_i}$ are given. Then, $\oPot_{D_i}$ and $\oDist_{\frac{1}{\log^3(n)},D}$ can be efficiently computed for every $D \in \left[ \frac{D_i}{\valSNC},D_i \right]$.
\end{restatable}
\cref{lem:cyclethree} follows mainly from previous work. More precisely, having access to a $\tO(1)$-competitive $\ell_1$-oblivious routing in $G_{S,D_i}$, we can compute the following two objects via boosting and rounding \cite{She17, BKKL17, Li20, goranci2022universally}. First, a $(1+\eps)$-SSSP-tree in $G_{S,D_i}$ rooted at $\sstar$. Second, an individually good $(1+\eps)$-approximate potential in $G_{S,D_i}$ for the single-source transshipment demand with source $\sstar$. 
For completeness, we give a summary of these steps in the proof of the full version of the paper.

If one looks at what the aforementioned tree and potential in $G_{S,D_i}$ correspond to in the graph $G$, then one can relatively straightforwardly transform them to obtain a $\left(1+\frac{1}{\log^3 n}\right)$-approximate forest for $D$ rooted at $S$ and a potential for scale $D_i$ with respect to $S$, assuming $\eps = \frac{1}{\poly(\log n)}$ is sufficiently small.
The details of this transformation can be found in the full version of the paper.

\begin{restatable}{lemma}{cyclefour}
\label{lem:cyclefour}
Assume oracle $\oDist_{\frac{1}{\log^3(n)},D}$ is given for every $D \in \left[ \frac{D_i}{\valSNC},D_i \right]$. Then, a sparse neighborhood cover with covering radius $\frac{D_i}{\valSNC}$ together with a rooted spanning tree $T_C$ of diameter at most $D_i$ for every cluster $C$ in the cover can be computed efficiently.
\end{restatable}

\cref{lem:cyclefour} directly follows from \cite[Theorem C.4]{elkin2022decompositions}, which is proven by adapting the algorithms of \cite{rozhon_ghaffari2019decomposition,chang_ghaffari2021strong_diameter} for the closely-related network decomposition problem. 


\begin{restatable}{lemma}{cyclefive}
\label{lem:cyclefive}
Assume we are given an oracle $\oPot_{D_i}$ and a sparse neighborhood cover with covering radius $\frac{D_i}{\valSNC}$ together with a rooted spanning tree $T_C$ of diameter at most $D_i$ for every cluster $C$ in the cover. 
Then, a distance structure for scale $D_i$ can be computed efficiently.
\end{restatable}

Given a sparse neighborhood cover together with a tree for each cluster, it only remains to compute the potential for scale $D_i$ with respect to $V \setminus C$ for each cluster $C$ in the sparse neighborhood cover. 
One can compute the potentials for all the clusters in a given clustering $\fC$ simultaneously. The simplest approach for computing these potentials is to compute a single potential for scale $D_i$ with respect to all the nodes that are neighboring one of the clusters in $\fC$. This approach works as long as there does not exist a node that is both clustered and neighboring a different cluster. Even though this can indeed happen, there is a simple solution that solves this problem. The details can be found in the proof of \cref{lem:cyclefive} in the full version of the paper. 

\alert{this lemma is never proven, what is proven is a (differently stated and more precise) Theorem 5.1.. we should probably replace this lemma with that theorem.}
\begin{restatable}{lemma}{cycleeuler}
\label{lem:cycleeuler}
Assume the oracle $\oEuler$ is given. Then, $\oRound$ can be efficiently computed. 
\end{restatable}
The main ideas to prove this lemma were already discussed in the introduction.
The details can be found in the full version of the paper.

\subsection{Main Theorems}
\label{sec:main_theorems}

In this part, we state the main theorems of this paper. 

First of all, a simple induction proof on top of \cref{lem:cycleone} leads to the following result.

\begin{restatable}{lemma}{cyclesix}
\label{lem:cyclesix}
Assume oracle $\oRound$ is given. Then, a distance structure for every scale $D_i$ can be efficiently computed.
\end{restatable}

Note that for $i = 1$, \cref{lem:cycleone} states that given access to $\oRound$, one can efficiently compute a distance structure for scale $D_1$.
The complete induction proof can be found in the full version of the paper.

Given access to a distance structure for every scale $D_i$, the theorem below follows directly from our $\tilde{O}(1)$-competitive $\ell_1$-oblivious routing scheme described in the full version of the paper.

\begin{restatable}[Distance Structures give $\ell_1$-Oblivious Routing]{theorem}{cycleseven}
\label{theorem:cycleseven}
Assume a distance structure for every scale $D_i$ is given. Then, there exists a $\tilde{O}(1)$-competitive $\ell_1$-oblivious routing $R$ for $G$ for which $R$ and $R^T$ can be efficiently evaluated.
\end{restatable}

As discussed in \Cref{sec:intro}, $\ell_1$-oblivious routing and rounding is sufficient to solve the $(1+\eps)$-SSSP tree and the $(1+\eps)$-transshipment problems~\cite{She17, BKKL17,Li20,goranci2022universally}.

\begin{restatable}[$\ell_1$-Oblivious Routing gives SSSP and transshipment]{theorem}{cycleeight}
\label{theorem:cycleeight}
Assume oracle $\oRound_{\eps/2}$ is given for some $\eps \in (0,1]$ and that there exists an efficient algorithm to evaluate $R$ and $R^T$ for some $\tilde{O}(1)$-competitive $\ell_1$-oblivious routing $R$ for $G$. Then, the $(1+\eps)$-transshipment problem and the $(1+\eps)$-SSSP-tree problem in $G$ can be solved in $\tilde{O}(1/\eps^2)$ \aggregate rounds. 
\end{restatable}

Combining \cref{lem:cyclesix}, \cref{theorem:cycleseven} and \cref{theorem:cycleeight} results in the following theorem.

\begin{restatable}{theorem}{cyclenine}
\label{theorem:cyclenine}
Assume oracle $\oRound_{\eps/2}$ is given for some $\eps \in (0,1]$. The $(1+\eps)$-transshipment problem and the $(1+\eps)$-SSSP-tree problem in $G$ can be solved in $\tilde{O}(1/\eps^2)$ \aggregate rounds.
\end{restatable}

The theorem above together with the fact that $\oRound$ can be efficiently implemented given $\oEuler$ implies the following result.

\begin{restatable}{theorem}{cycleten}
\label{theorem:cycleten}
Assume oracle $\oEuler$ is given and let $\eps \in (0,1]$. The $(1+\eps)$-transshipment problem and the $(1+\eps)$-SSSP-tree problem in $G$ can be solved in $\tilde{O}(1/\eps^2)$ \aggregate rounds.
\end{restatable}

The \euler problem can be solved with near-linear work and polylogarithmic depth \cite{atallah_vishkin1984euler_pram}. Together with the theorem above and the fact that each $\aggregate$ round can be simulated with near-linear work and polylogarithmic depth, we obtain our main parallel result.

\mainparallel*

Moreover, our \congest algorithms for the \euler problem developed in the full version of the paper together with general simulation results for the \congest model developed in prior work, we obtain our main result in the \congest model.

\maindistributed*

We finish this section by stating the conditional results on universally-optimal SSSP and transshipment algorithms one can obtain from this work:

\begin{theorem} \label{lem:universalOptreduction1}
Suppose there exists a deterministic algorithm for partwise aggregation that runs in $\shortcutQuality{G} \cdot n^{o(1)}$ \congest rounds, then there exist deterministic $(1+\eps)$-SSSP and $(1+\eps)$-transshipment algorithms with a round complexity of $\shortcutQuality{G}) \cdot n^{o(1)}$, which is universally-optimal up to a $n^{o(1)}$-factor. 
\end{theorem}

We get even stronger conditional results if more efficient \congest algorithms for computing cycle covers as defined in \cite{parter_yogev2019cycle_decomp_near_linear,parter_yogev2019cycle_covers_minor_closed} are given. 

\begin{theorem}\label{lem:universalOptreduction2}
Suppose there exists a deterministic algorithm for partwise aggregation that runs in $\tO(\shortcutQuality{G})$ \congest rounds and a $(\tO(1),\tO(1))$ cycle cover algorithm for $\tO(1)$-diameter graphs which runs in $\tO(1)$ \congest rounds. Then, there exist deterministic $(1+\eps)$-SSSP and $(1+\eps)$-transshipment algorithms with a round complexity of $\tO(\shortcutQuality{G})$, which is universally-optimal up to polylogarithmic factors. 
\end{theorem}

While the polylogarithmically tight algorithmic results assumed in \Cref{lem:universalOptreduction2} seem out of reach of current techniques, our conditional results show that the problem-specific part towards universally optimal shortest path algorithms, even deterministic ones, are essentially fully understood through the techniques of this paper.

\bibliographystyle{alpha}
\bibliography{Refs}

\end{document}